\documentclass[10pt,journal,comsoc]{IEEEtran}
\ifCLASSOPTIONcompsoc
  \usepackage[nocompress]{cite}
\else
  \usepackage{cite}
\fi

\ifCLASSINFOpdf
   \usepackage[pdftex]{graphicx}
\else
   \usepackage[dvips]{graphicx}
\fi
\graphicspath{{./plots/}{./figs/}}
\usepackage{amsmath}

\usepackage{caption}
\ifCLASSOPTIONcompsoc
  \usepackage[caption=false,font=footnotesize,labelfont=sf,textfont=sf]{subfig}
\else
  \usepackage[caption=false,font=footnotesize]{subfig}
\fi
\usepackage{amsmath,bm}
\usepackage{subfig}

\usepackage{xcolor}
\usepackage{booktabs}
\usepackage{multirow}
\usepackage{paralist}
\usepackage{amsfonts}
\usepackage{dsfont}
\usepackage{mathtools}
\usepackage[inline]{enumitem}

\usepackage{amsthm}
\usepackage{algorithm}
\usepackage[noend]{algpseudocode}

\newtheorem{theorem}{Theorem}[section]

\newtheorem{prop}{Proposition}[section]

\theoremstyle{definition}

\theoremstyle{remark}

\theoremstyle{remark}

\usepackage{url}

\newcommand{\etal}{\textit{et al}. }

\newcommand{\ie}{\textit{i}.\textit{e}.}
\newcommand{\eg}{\text{e}.\textit{g}.}

\DeclareMathOperator*{\argmax}{arg\,max}
\DeclareMathOperator*{\argmin}{arg\,min}
\newcommand{\CP}{{\textsf{\footnotesize CP}}}
\newcommand{\ANO}{{\textsf{\footnotesize ANO}}}

\def\NoNumber#1{{\def\alglinenumber##1{}\State #1}\addtocounter{ALG@line}{-1}}

\begin{document}
%\title{Content provider subsidies to realize the optimal memory bandwidth tradeoff 
\title{Cache Subsidies for an Optimal Memory for Bandwidth Tradeoff in the Access Network
}

\author{Mahdieh Ahmadi,
        James Roberts,~%\IEEEmembership{Fellow,~OSA,}
        Emilio Leonardi,~%\IEEEmembership{Fellow,~OSA,}
        and~Ali Movaghar,~\IEEEmembership{Senior~Member~IEEE}% <-this % stops a space
\IEEEcompsocitemizethanks{
\IEEEcompsocthanksitem M. Ahmadi and A. Movaghar are with the Department of Computer Engineering, Sharif University of Technology, Tehran, Iran. (e-mail:mahmadi@ce.sharif.edu, movaghar@sharif.edu)
\IEEEcompsocthanksitem J. Roberts is with the Network and Computer Science Department, Telecom ParisTech, Paris, France. (e-mail:jim.roberts@telecom-paristech.fr)
\IEEEcompsocthanksitem E. Leonardi is with the Politecnico di Torino, 10129 Turin, Italy (e-mail:emilio.leonardi@polito.it)}
%of Electrical and Computer Engineering, Georgia Institute of Technology, Atlanta,
%GA, 30332.\protect\\
%E-mail: see http://www.michaelshell.org/contact.html
%\IEEEcompsocthanksitem J. Doe and a2 are with Anonymous University.}% <-this % stops an unwanted space
}

\IEEEtitleabstractindextext{%
\begin{abstract}
While the cost of the access network could be considerably reduced by the use of caching,  this is not currently happening because content providers (CPs), who alone have the detailed demand data required for optimal content placement, have no natural incentive to use them to minimize access network operator (ANO) expenditure. We argue that ANOs should therefore provide such an incentive in the form of direct subsidies paid to the CPs in proportion to the realized savings. We apply coalition game theory to design the required subsidy framework and propose a distributed algorithm, based on Lagrangian decomposition, allowing ANOs and CPs to collectively realize the optimal memory for bandwidth tradeoff. The considered access network is a cache hierarchy with per-CP central office caches, accessed by all ANOs, at the apex, and per-ANO dedicated bandwidth and storage resources at the lower levels, including wireless base stations, that must be shared by multiple CPs. 
\end{abstract}

}

\maketitle
\IEEEdisplaynontitleabstractindextext  
\IEEEpeerreviewmaketitle

\section{Introduction}
 %Mobile data traffic will increase sevenfold between 2017 and 2022. It is estimated that this traffic will be dominated by 82\% with video content in 2022 \cite{cisco2018cisco}. To cope with this traffic shift, the MNO can deploy hard infrastructure from the city centre to the far end to offer a distributed shared cache and lesson the investment for bandwidth expansion. However realizing the optimal memory for bandwidth tradeoff is not trivial for the MNO. On one hand, the investment on bandwidth must be made in advance of applications (e.g. Content Providers (CPs), Content Delivery Networks(CDNs) and Over The Air (OTA)) being ready to take advantage of it due to high lead time. On the other hand, the ongoing tussle between MNO and these applications, the optimal cache dimension, can not be realized by the MNO. 

Since the advent of the Web and the explosive expansion of the Internet, network operators have recognized the potential for significantly reducing infrastructure costs by locally caching copies of popular contents rather than fetching them repeatedly from a remote source. However, despite the clear and increasing economic advantages of caching, it remains true that content is rarely stored at any site closer to end users than points of presence (PoPs) situated well upstream of the access network.

Network operator caching ambitions were initially stymied by more effective competition from global reach content distribution networks (CDNs) like Akamai. More significantly, content providers (CPs), like Google, Facebook and Netflix who came later and are now the source of most Internet traffic, have developed lucrative business models that rely on exclusive knowledge of their respective customer base. They preserve this knowledge by encrypting transmissions and thus prevent operators from transparently caching their contents. CPs do cache popular content in ISP PoPs both to reduce their costs and to improve end user quality of experience. However, they have hardly any additional incentive to move caches even closer to users within the access network. It is the main thesis of this paper that network operators need to create such an incentive by financially rewarding CPs for content placements that reduce infrastructure costs. The gain from optimizing the memory for bandwidth tradeoff realized by caching is considerable and network operators, CPs and ultimately end-users will all benefit significantly. 

  \begin{figure}[t]
 \centering
 \captionsetup{justification=centering}
 \includegraphics[width=0.48\textwidth]{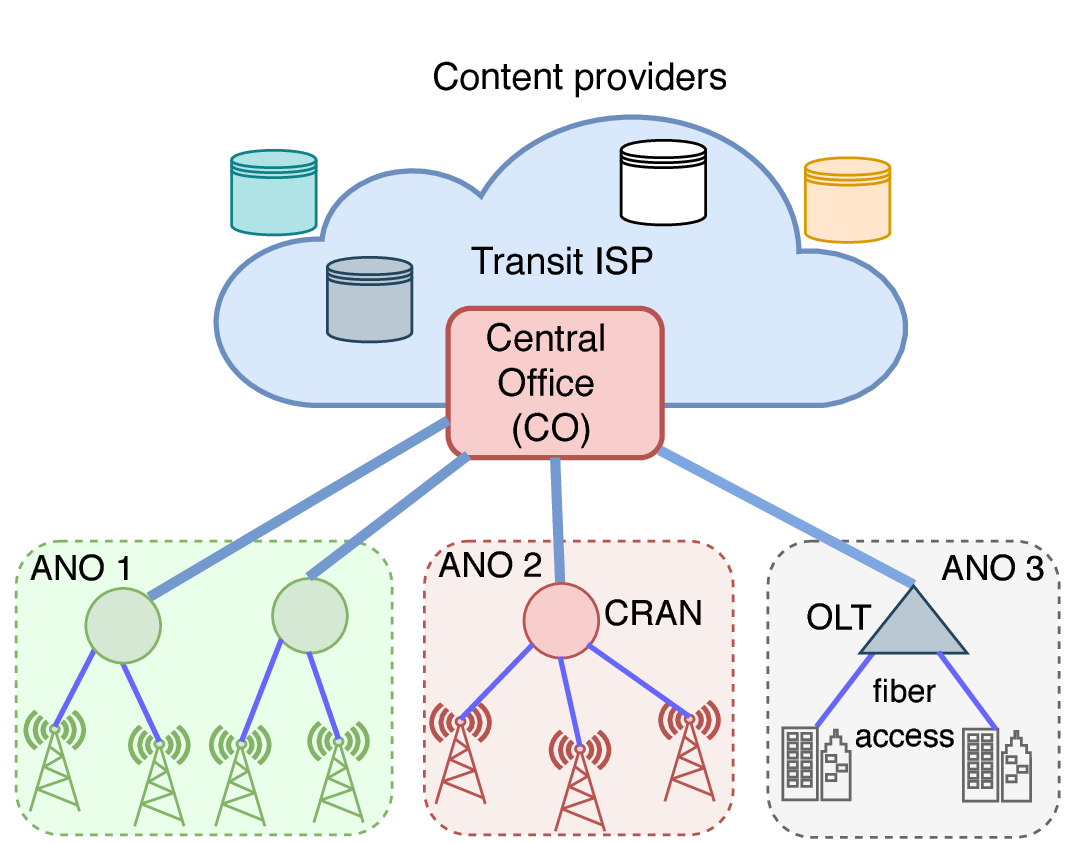}
 \caption{Access network topology}
 \label{fig:network}
 \end{figure}  
 
We consider the simple tree-shaped access network topology depicted in Fig. \ref{fig:network}. The root of the tree is a central office (CO) equipped as a datacenter \cite{peterson2016}. This is a hub connecting multiple access network operators (ANOs) to the Internet via an ISP. An ANO might, for instance, be a mobile network operator with base stations (BSs) at the lowest layer, or a fixed network operator with optical network units (ONUs) shared by users in the same building. These ANOs may have intermediate nodes housing a cloud-radio access network (C-RAN) or an optical line terminal (OLT), respectively. These intermediate nodes are assumed to have (micro-)datacenter capabilities, like the CO,  enabling flexible storage provision while base stations have just a limited capacity cache memory such as a solid-state disk. All ANOs deliver contents from all CPs and demand is assumed independent of content placements. This network model is clearly not completely general but is intended to be sufficiently generic to illustrate the challenges raised by the design and implementation of cache subsidies.

The use of subsidies to drive an optimal memory for bandwidth tradeoff in the access network is, to the authors' knowledge, an original proposal.  The main contributions of the paper are as follows:
\begin{itemize}
\item the potential savings to be gained from an optimal tradeoff are roughly quantified suggesting that they may count in billions of dollars per year for a large network operator;
\item the need for CP subsidies as a catalyst for realizing these savings is identified as a consequence of the powerful position of the CPs in the ANO two-sided market;
\item coalition game theory is applied to determine how the cost of a CP cache in the CO should be shared between the ANOs  using it and how much they contribute to the subsidy;
\item a distributed scheme is proposed to optimize the access network tradeoff where CPs optimally place content based on unit storage and bandwidth prices while these prices are modulated by ANOs as necessary to meet capacity constraints.
\end{itemize}
%\blue{this can be moved before the contribution section and the explanation of the rest of the paper put here?}
The following four sections successively present these contributions. We discuss related work in the penultimate section before drawing conclusions and highlighting directions for further study.
\section{Memory for bandwidth tradeoff}\label{sec:mbt}
This section discusses the memory bandwidth tradeoff for an access network taking the form of a tree, assuming known demand and given storage and bandwidth costs.

 \subsection{Cache performance}\label{sec:mbt:cp}
We consider a catalogue of contents $\mathcal{F}$ with a large number $F$ of items (`files') that, to simplify but without loss of generality, we assume to be of constant size. 
%As $F$ is large and items are individually relatively small we can suppose without loss of generality that items are of unit size and cache performance is a function of the relative size of cache and catalogue expressed as a number of contents \cite{elayoubi2015performance}. 
To estimate cache hit probabilities, we assume requests for each content $f$ 
generate traffic at rate $\lambda^f$ content downloads per second, proportional to a popularity law $q^f$, \ie\, $\lambda^f = T q^f$ where $\sum_{f\in \mathcal{F}}{q^f}=1$, $q^1\ge q^2 \ge \cdots \ge q^F$ and $T$ is the overall traffic demand in downloads per second. $T$ here is the peak traffic used for network dimensioning.

For the present work we assume the performance of a cache of size $C$ is ideal, yielding a hit probability $h(C)=\sum_{f \le C} q^f$. This would be realized by proactive placement of the most popular items. Note that reactive caching policies like \emph{least recently used} (LRU) are not efficient in the access network where demand is low relative to the rate of content churn \cite{elayoubi2015performance}.

Performance depends significantly on the popularity law. Observations consistently show that request rates for the most popular items roughly follow a Zipf law, $q^f \propto f^{-\alpha}$, with $\alpha~\approx~0.8$ (\eg, \cite{imbrenda2014analyzing}). There is however a very large amount of content in the tail of the distribution whose popularity is imperfectly estimated but is much smaller than would be predicted by the Zipf law \cite{olmos2016, roberts2013exploring}. This high volume content, like old photos posted on Facebook or family videos shared on YouTube, is stored in CP datacenters but is hardly ever cached. For the evaluations in this paper, we generally assume the catalogue $\mathcal{F}$ does follow a Zipf law with a \emph{useful} size $F$ corresponding to a volume of data between 1 TB and 1 PB, it being implicitly assumed that this only represents the most popular contents.

\subsection{Costs of storage and bandwidth}\label{sec:mbt:cost}
Caching trades off the cost of storage for the cost of bandwidth. While the cost of storage is well-known, it proves difficult to obtain accurate estimates of the bandwidth cost. This should be known to the network providers realizing the tradeoff but for the present paper we are obliged to make simplifying assumptions.

The cost of bandwidth is the cost of all transport and routing equipment dimensioned to carry busy period traffic between given points in the network with adequate quality of service. We assume this cost is proportional to peak demand. Costs clearly depend on the particular network instance in question but for the present discussion we consider a fixed cost per unit of busy hour demand. This might, for instance, be the long run average incremental cost (LRIC) used by telecommunications regulators.

A bandwidth cost used in prior work was derived from the marginal cost of backhaul charged to mobile network operators in France at \$2 per Mb/s per month \cite{elayoubi2015performance}. An alternative estimate derives from the cost of bandwidth charged by Cloud providers like Google \cite{googleCloud}. Catalogue GB per month download prices convert to a monthly rate of around \$4 per Mb/s of busy hour traffic\footnote{Assuming busy hour demand is $1/8^{th}$ of daily traffic.}. 

A rough cost of storage can be deduced from the rates charged by Cloud providers. A monthly rate of around \$0.03 per GB is the current offer for frequent access storage \cite{googleCloud}. 
%\begin{table}[t]
%  \caption{Frequently Used Notation.}
%  \label{table:notation}
%  \centering
%  \begin{tabular}{l p{0.8 \columnwidth} }
%    \hline 
%    Symbol & Description \\ \hline \hline
%    $k \in {\mathcal K}$ & CP index\\ \hline
%    $a \in {\mathcal A}$ & ANO index\\ \hline
%    $f \in \mathcal F$ & file index\\ \hline
%    $n \in \mathcal N$& cache node index \\ \hline
%    $l \in \mathcal L$& leaf index \\ \hline
%     $i \in \mathcal I$& intermediate node index \\ \hline
% %  $i \in \mathcal I$& Intermediate node index \\ \hline
%     $\mathcal F_k$ &catalog of CP $k$ with size $F_k=|\mathcal F_k|$\\ \hline
%     $r_{ak}$ &  CP $k$ share of ANO $a$ cost saving\\ \hline
%     $\lambda_{n}^f$ & demand for content $f$ at node $n$\\ \hline
%    $T_a$ &  traffic demand from ANO $a$ users\\ \hline
%    $q_{a}^f$ & popularity law of ANO $a$ demand over $\mathcal F_k$ \\ \hline
%    $C_{nk}$ & cache size of CP $k$ at node $n$\\ \hline
%     $\Lambda_{nk}$ & residual demand of CP $k$ routed over link $n$\\ \hline
%    $S_n$ & cache size limit at node $n$\\ \hline
%    $B_n$ &  bandwidth limit on link $n$\\ \hline
%    $s_n$& per-content cost of storage at node $n$ \\ \hline
%    $b_n$& per-request/s cost of bandwidth between node $n$ and its parent\\ \hline
%    $h(C)$&hit probability of cache of size $C$ \\ \hline
%   \end{tabular}
% \end{table}
\begin{table}[t]
  \caption{Frequently Used Notation.}
  \label{table:notation}
  \centering
  \begin{tabular}{l p{\ifCLASSOPTIONtwocolumn 0.7 \columnwidth \else 0.45 \columnwidth \fi} }
    \hline 
    $k \in {\mathcal K}$ & CP index\\ 
    $a \in {\mathcal A}$ & ANO index\\ 
    $f \in \mathcal F$ & file index\\ 
    $n \in \mathcal N$& cache node index \\ 
    $l \in \mathcal L$& leaf index \\ 
     $i \in \mathcal I$& intermediate node index \\ 
 %  $i \in \mathcal I$& Intermediate node index \\ 
     $\mathcal F_k$ &catalog of CP $k$ with size $F_k=|\mathcal F_k|$\\ 
     $r_{ak}$ &  CP $k$ share of ANO $a$ cost saving\\ 
     $\lambda_{n}^f$ & demand for content $f$ at node $n$\\ 
    $T_a$ &  traffic demand from ANO $a$ users\\ 
    $q_{a}^f$ & popularity law of ANO $a$ demand over $\mathcal F$ \\ 
    $C_{nk}$ & cache size of CP $k$ at node $n$\\ 
     $\Lambda_{nk}$ & residual traffic of CP $k$ routed over link $n$\\ 
    $S_n$ & cache size limit at node $n$\\ 
    $B_n$ &  bandwidth limit on link $n$\\ 
    $s_n$& cost of storage at node $n$ \\ 
    $b_n$& cost of bandwidth between node $n$ and its parent\\ 
    $h(C)$&hit probability of cache of size $C$ \\ \hline
   \end{tabular}
 \end{table}
 
\subsection{Realizing the optimal tradeoff}\label{sec:ufl}
We model an access network bringing download traffic to users as a tree network with facility for cache storage at each of its nodes (Fig. \ref{fig:3tier}). The network has a set of $\mathcal N$ nodes with root at node  $0$, the CO. Caches at the lowest tier (\eg, base stations) constitute the set of leaf nodes $\mathcal{L}$. 
%Node $n \notin \mathcal{L}$ has children $\mathcal{E}(n)$ and the parent of node $n$ is $p(n)$ for $n> 0$. For convenience we additionally define the parent of node 0 to be the source and denote this by index +, i.e., $p(0)=+$. 
Demand for content $f$ from leaf node $l$ is $\lambda_l^f$ downloads/s. 
%For other nodes we recursively define cumulative demands $\lambda_n^f = \sum_{m\in \mathcal{E}(n)} \lambda_m^f$. 
The unit cost of bandwidth between source and root is $b_{0}$ and the cost between any other node $n$ and its parent is $b_{n}$. %Cumulative bandwidth cost from leaf $l$ to nodes  $n \notin \mathcal{L}$ is denoted $b_{ln}$. This is determined by  $b_{ln} = \sum_{l \preceq m \prec n} b_{m}$. Bandwidth cost from $n$ to the source is $b_{n+} = \sum_{n \preceq m \preceq 0}b_{m}$. 
Unit storage cost at node $n$ is $s_n$.

%starting with $m \in \mathcal L$
%To place each content  $f$ we must solve an uncapacitated facility location (UFL) problem on a tree. The UFL solution specifies the optimal set of nodes $\mathcal{S}$ which minimizing overall cost:
The placement which minimizes the total storage and bandwidth cost can be determined independently for each content. To place each content  $f$ we must solve an uncapacitated facility location (UFL) problem on a tree. The UFL solution specifies the optimal set of nodes $\mathcal{S}$ which minimizes the overall cost for the given content: % and determines the corresponding minimum cost  
\begin{equation}\label{eq:ufl}
\mathcal{S} = \argmin_{\mathcal S \subseteq \mathcal N} \bigg(\sum_{n \in \mathcal{S}} s_n + \sum_{l \in \mathcal{L}} \lambda_l^f \min_{n \in \mathcal{S}} \sum_{l \preceq m \prec n} b_m\bigg),
\end{equation}
where the summation range denoted ${l \preceq m \prec n}$ covers the set of links in the unique path from $l$ to $n$.

The literature provides algorithms to solve this problem, \eg \,Cornuejols \etal \cite{cornuejols1983uncapicitated} has an $ O(N^2)$ algorithm while Shah \etal \cite{shah2002undiscretized} improves this to $ O(N \log N)$. The UFL algorithm should be applied successively for contents in decreasing order of overall demand, $\sum_{l \in \mathcal{L}} \lambda_{l}^f$. Application can cease when placement of at least one locally more popular content has previously been refused at every node. 
 \begin{figure}[t]
 \centering
 \hspace{-10mm}
 \includegraphics[width=0.5\textwidth]{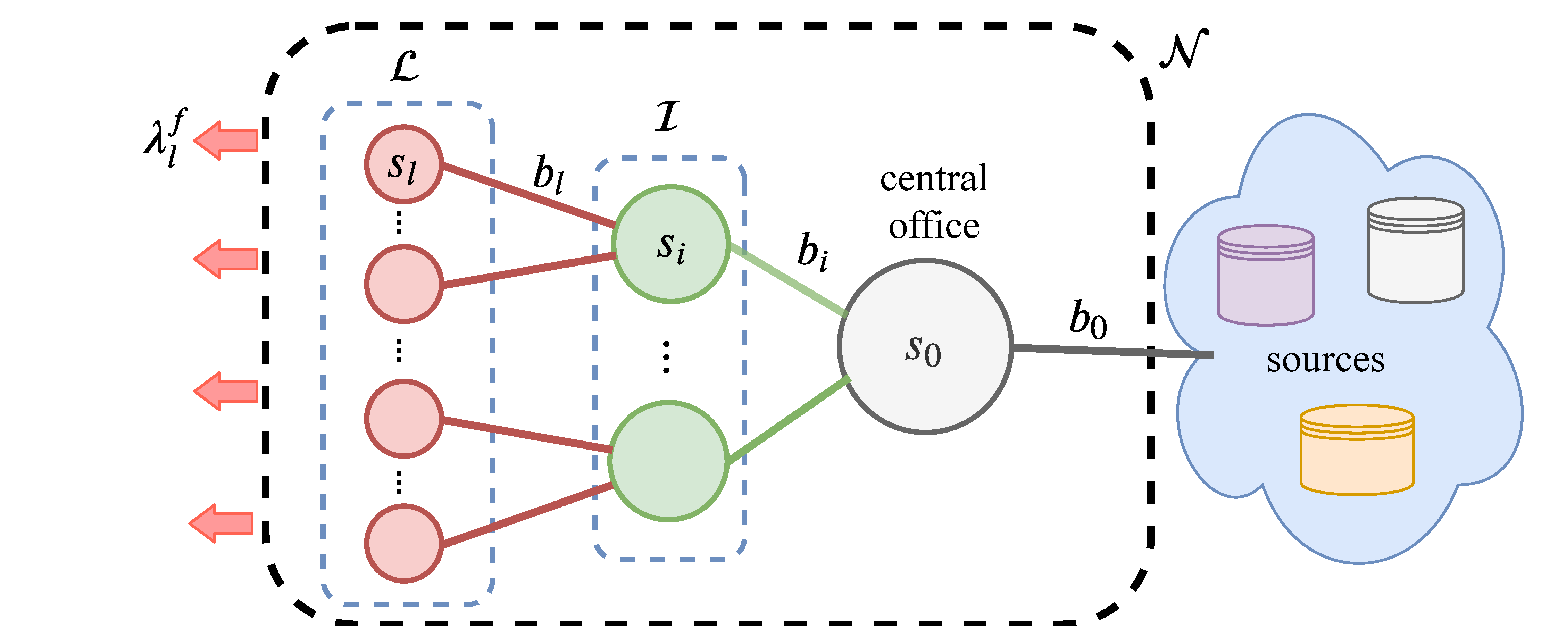}
 \caption{Three tier access network with nodes $\mathcal{N}$, including root $0$, leaves $\mathcal{L}$ and intermediate nodes $\mathcal{I}$ with storage cost $s_n$ and bandwidth cost $b_n$.}
  \label{fig:3tier}
 \end{figure}  
 \subsection{Quantifying the tradeoff}\label{sec:quantify}
 While it is in general necessary to perform optimal placement using an UFL algorithm, to gain intuition we adopt a simpler network model in this subsection. 
 %, following \cite{kelly2001optimal}. 
 We consider a symmetric 3-tier tree where demand at each leaf is statistically identical and total demand is $T$ over all the leaves. The root at tier $3$ has $e_2$ tier $2$ children each of which has $e_1$ tier $1$ children (the leaves). 

As catalogues and cache sizes are large, it is reasonable for the present evaluation to consider $F$ and $C$ as real variables and to reason in terms of bytes rather than discrete contents \cite{elayoubi2015performance}. The hit probability for a Zipf$(\alpha)$ popularity law is approximated by $h(C)=(C/F)^{1-\alpha}$, derived on replacing summations $\sum_{f\le C} q^f$ by an integral. Monthly storage costs are $s^{(i)}$ at tier $i$ nodes, \ie, a cache of size $C$ costs $Cs^{(i)}$ per month. Bandwidth cost to bring content from the next highest tier to tier $i$ (or from source to root when $i=0$) is denoted $b^{(i)}$, \ie, the cost to transmit demand $T$ is $Tb^{(i)}$ per month. 

Assuming storage is possible at each node, the minimum cost network is such that tier $1$ nodes cache the $C_1$ most popular bytes of content, tier $2$ nodes cache the next $C_2$ most popular while the root caches the next $C_3$, for some values of $C_1$, $C_2$ and $C_3$. Network cost is then,   
%\red{since $C$ is the number of files, we need the size of an object here?} JR: s is defined per content
\begin{multline*}
\mathsf{cost} =  e_1e_2C_1s^{(1)}+e_2C_2s^{(2)} +C_3s^{(3)}  +T \left( \left(1-h(C_1)\right) b^{(1)} +\right. \\ \left. + \left(1-h(C_1+C_2)\right) b^{(2)}+\left(1-h(C_1+C_2+C_3)\right) b^{(3)} \right).
\end{multline*}
For practically relevant parameter values, standard methods yield optimal cache sizes:
\begin{align}\label{eq:cdependent}
C_1 &= F  \times \min\bigg\{1, \bigg(\frac{(1-\alpha)Tb^{(1)}}{F e_2(e_1s^{(1)}-s^{(2)})}\bigg)^{1/\alpha}\bigg\},\notag\\
C_1+C_2 &= F  \times  \min\bigg\{1, \bigg(\frac{(1-\alpha)Tb^{(2)}}{F(e_2s^{(2)}-s^{(3)})}\bigg)^{1/\alpha}\bigg\},\notag\\
C_1+C_2+C_3 &=  F  \times \min\bigg\{1,  \bigg(\frac{(1-\alpha)Tb^{(3)}}{Fs^{(3)}}\bigg)^{1/\alpha}\bigg\}.
\end{align}
%\red{There should be condition on $e_1s^{(1)}-s^{(2)} >0$ and $e_2s^{(2)}-s^{(3)} > 0$ and if $b^{(2)} << b^{(1)}$, there may happen that $C_1+C_2 < C_1$. So we should also calculate the maximization with respect to the previous level cache or put some condition on $b$ values.}

Optimal cache sizes when storage is not possible at one or two tiers can be similarly derived.  Fig.~\ref{fig:compall} plots comparative gains for a network with parameter values, given in the caption, that are meant to be representative of a mobile access network \cite{elayoubi2015performance}.  Savings are plotted against the cost factor $\Gamma=(Tb^{(i)})/(Fs^{(i)})$, that here summarizes the impact of the individual parameters. The figure compares savings possible with caches at all tiers, caches at tiers $1$ and $3$, caches at tiers $1$ and $2$ and caches at tier $1$ only.
  \begin{figure}[t]
  	\centering
        \subfloat[$e_1=100, e_2=10$]{
       \includegraphics[width=0.48\columnwidth]{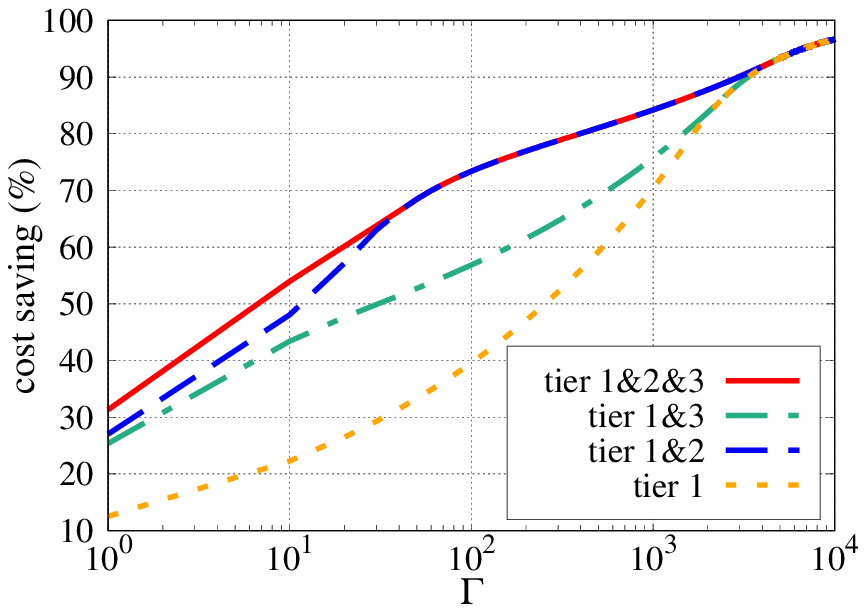}
         \label{fig:comp1}}\hspace{-0.12in}
          \subfloat[$e_1=10, e_2=100$]{
         \includegraphics[width=0.48\columnwidth]{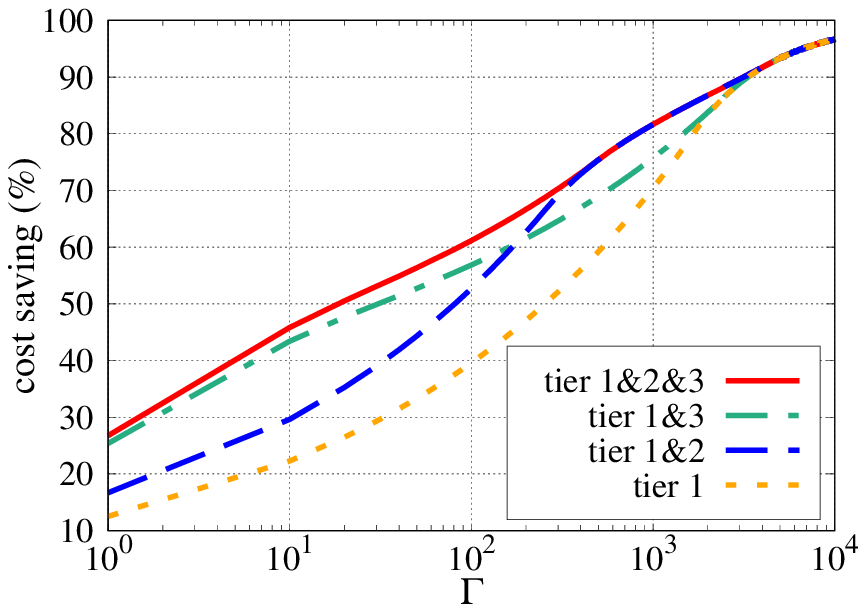}
         \label{fig:comp2}}
        \caption{Access network cost savings against cost factor $\Gamma=(Tb^{(i)})/(Fs^{(i)})$; \ifCLASSOPTIONonecolumn \\ \fi $b^{(1)}=b^{(2)}=b^{(3)}=\$4$ per Mb/s per mth, $s^{(1)}=s^{(2)}=s^{(3)}=\$0.03$ per GB per mth, Zipf($0.8$) popularities.} 
        \label{fig:compall}
    \end{figure}

A first remark is that the savings are potentially very large, especially for high $\Gamma$. For instance, a total demand volume of $10$ Gb/s from 1000 BSs for a CP catalogue volume of 10 TB corresponds here to $\Gamma= 133$ and an optimal saving of more than $70\%$ for the data of Fig. \ref{fig:comp1}. Data gathered by Perillo \etal \cite{perillo2016panel} provide a reality check on the size of the savings possible through an optimal tradeoff: the annual revenue of network operators is measured in tens of billions of dollars (\eg, \$$130$B for AT\&T in 2013) and CapEx is roughly $15\%$ of this. Since most expenditure is for the access network, we estimate that potential annual savings are measured in billions of dollars.

Fig. \ref{fig:comp1} shows that most savings in this example come from caches at the two lower tiers, 1 and 2, while tier 1 BS caches alone are generally inadequate. When fanout values are inversed in Fig. \ref{fig:comp2}, however, a cache at tier $3$ (the CO) brings greater savings than a cache at tier $2$ for $\Gamma < 10^2$. In general, caches at all three layers significantly contribute to the overall savings.

\section{The content delivery business }\label{sec:CDN}
We discuss the business environment arguing ANOs must subsidize CPs in order to realize the optimal memory for bandwidth tradeoff.   
\subsection{Content providers}
Major content providers like Netflix, YouTube and Facebook currently generate the vast majority of Internet traffic \cite{sandvine}. These CPs represent multi-billion dollar businesses gaining considerably more revenue and having greater market power than the network operators on which they rely for content delivery. We include CDNs like Akamai among the CPs as they also manage large volumes of traffic and have similar business relations with the ANOs. 

The CP business model generally relies on exclusive and deep knowledge of customer behavior, used for recommendation systems, ad placement and other strategic marketing activities. Customer behavior is also rightly considered to be highly confidential and for this reason alone would not be shared with any network operator seeking to realize the memory bandwidth tradeoff.  The significant gains highlighted in the previous section must be realized therefore by exploiting the CPs extensive knowledge of how customer demand is distributed over its content catalogue while preserving the exclusivity of this knowledge. 

\subsection{ A two-sided market}\label{sec:CDN:2sided}
An Internet access network, viewed as a platform, is a two-sided market where end-users constitute the ``money side'' and CPs the ``subsidy side'': end-users consume content and pay the ANO for connectivity while CPs supply the content and typically do not pay anything to the ANO for the traffic this generates \cite{rochet2006two}. The absence of compensation from the CPs is frequently the source of ANO complaints and has generated heated discussions about network neutrality  (e.g., Comcast versus Netflix in the US, Free versus Google in France). To incite compensation, the ANOs might be tempted to provide less than adequate bandwidth to carry generated traffic by non-paying CPs but the impact on quality  mainly hurts their paying end-user customers. CPs typically do pay the ISPs to which their servers connect but these are usually distinct from the considered ANOs. Proposals for sharing the revenue between ``CP ISPs'' and ``eyeball ISPs'' have, to our knowledge, never been implemented \cite{ma2010cooperative}.

Of course, CPs do employ caching, typically placing dedicated servers in one or a small number of PoPs in the ISP network upstream of the CO and the access networks considered here \cite{bottger2018open, yap2017taking, wohlfart2018leveraging}. This is advantageous both to reduce the load of their datacenters and to improve customer QoE through lower latency. The propagation time is typically much smaller from the PoP than from the remote origin datacenter and has a significant impact on latency which is in turn very important for customer satisfaction. On the other hand, the small additional reduction in propagation time due to moving the cache from the PoP to any of the access network cache locations considered here has a negligible impact on perceived latency and QoE.  We believe enhanced QoE and its supposed impact on market share is therefore unlikely to motivate CP cooperation in realizing the optimal memory for bandwidth tradeoff.% \cite{ mitra2019case, gourdin2017economics}.

\subsection{Subsidized content placement}\label{sec:subsidy}

We claim ANOs need to persuade CPs to optimally place content in access network caches by sharing the resulting gain with them in the form of a direct subsidy. Consider, for illustration, an isolated cache dedicated to a given CP. The CP places content set $\mathcal{C}$ in the cache at total cost $C\times s$ where $C=|\mathcal{C}|$. End-users generate download traffic $T$ of which a proportion $h(\mathcal{C})$ is served by the cache. Bandwidth has unit cost $b$ so that the cache brings a net saving,
\begin{equation}
E(\mathcal{C}) = Th(\mathcal{C})\cdot b - C\cdot s. \label{eq:savings:1}
\end{equation}
The CP will have an incentive to cache the contents that maximize these savings if it receives a subsidy that is proportional to $E(\mathcal{C})$.  Both ANO and CP gain with this proposal. The remainder of this paper is about how such a subsidy might be realized in practice for the hierarchical access network discussed in Sec. \ref{sec:ufl}. We first introduce the network model. 

\subsection{Network model} \label{sec:netmodel}
The considered network topology is as depicted in Fig. \ref{fig:network}. It is a tree rooted on the central office, CO. The CO cache is shared by all the ANO instances connected to it while caches in downstream nodes, including base stations, are dedicated to their particular ANO instance. This topology is based on European Internet access. It is not perfectly general but usefully illustrates the main issues to be resolved. Extensions would be necessary, for instance, to account for ANO infrastructure sharing \cite{andrews2014}, or to more efficiently use base station caches when coverage areas overlap \cite{krolikowski2018decomposition}.

Multi-access edge computing (MEC) would be realized by the CO, intermediate nodes (CRANs or OLTs, say) or leaf nodes (base stations or ONUs), depending in particular on the latency requirements of the application in question. However, for \emph{content delivery}, the latency from a cache in any location would be sufficiently small and is not a placement criterion. 

We suppose the CO and transit ISP are owned by entities distinct from the ANOs, if necessary by imposed unbundling regulations. The former sell storage capacity and bandwidth at fixed unit rates and capacity is assumed unlimited. Other resources are assumed to belong to the specific ANO and their usage can be controlled by adjusting the ``prices'' $b$ and $s$ used to compute the CP subsidy. 

% For example, if storage capacity in the base station is limited by technological considerations, its usage could be controlled by adjusting $s$ relative to $b$. Bandwidth and storage are to some extent substitutable in that a lower price for the former reduces demand for the latter. Assuming bandwidth is over-provisioned based on forecast future traffic while storage capacity in the intermediate node is elastic \cite{kwak2018dynamic}, excess bandwidth can be filled by adjusting prices to reduce the relative attractiveness of storage. In addition to reducing overall network costs, the memory for bandwidth tradeoff thus offers new possibilities for optimizing the usage of installed capacity (see Sec. \ref{sec:optimization}).  
\section{Sharing a central office cache}\label{sec:shareCO}
We apply game theory to analyse value sharing between a CP and a set of ANOs who contribute to the cost of the CP cache at the central office. 

\subsection{A coalition game} \label{sec:game}
In our network model in Sec. \ref{sec:netmodel}, the CO hub is common to multiple ANOs. It provides cache space to CPs and is a gateway for IP transit. We analyse the memory for bandwidth tradeoff at the CO assuming ANOs are charged a common rate for transit bandwidth equivalent to $b$ per unit of peak traffic, expressed here in content downloads per second, and the CO storage charge is set to $s$ per content.  Available bandwidth and storage is  unlimited and the values of $b$ and $s$ are fixed. 
%content downloads per second

The cache used by each CP is distinct from that used by any other CP. This is required to preserve the privacy of customers and to ensure the integrity of the CP business models. We therefore consider a single generic CP providing content to multiple ANOs who must share the burden of the CP subsidy (cf. Sec. \ref{sec:subsidy}).   

To determine how savings due to caching should be shared, we apply  techniques from coalition game theory. 
The outcome of game $G(\mathcal{N},v)$ between players $n\in \mathcal{N}$ is determined by the value function $v(\mathcal{S})$ defined as the collective gain realizable by any sub-coalition $\mathcal{S} \subseteq \mathcal{N}$. The objective is to specify how the value should be distributed between players to maximize the gain while satisfying certain properties.  The value distribution, $\Phi:G \rightarrow \mathbb{R}^{|\mathcal N|}$, defines the amount allocated to each player. The following are two important concepts from coalition game theory \cite{saad2009coalitional}:
\begin{itemize}
\item the \emph{core} of game $G(\mathcal{N},v)$ is the set of all distributions $\Phi_n$, for $n \in \mathcal{N}$, such that no group of players would gain more by leaving to form a sub-coalition $\mathcal{S} \subset \mathcal{N}$,
\item the \emph{Shapley value} is the unique value distribution that satisfies four properties known as \emph{efficiency}, \emph{symmetry}, \emph{linearity} and \emph{null player}.
\end{itemize}

Denote the generic CP by $\CP$ and the set of ANOs by $\mathcal{A}$. We consider games over players $\{\CP \cup \mathcal{A}\}$ with value function such that $v(\mathcal{S})=0$ for any $\mathcal{S} \subseteq \mathcal{A}$ (i.e., any sub-coalition excluding $\CP$) and 
$$v(\CP \cup \mathcal{S}) = \max_{\mathcal{C}\subseteq \mathcal{F}} E(\mathcal{C}, \mathcal{S}),$$
where $E(\mathcal{C}, \mathcal{S})$ is the overall saving (cf. Sec.  \ref{sec:subsidy}),
\begin{equation}
E(\mathcal{C},\mathcal{S})=\sum_{a \in \mathcal{S}} (T_a h_a(\mathcal{C})\cdot b) - C \cdot s, \label{eq:savings}
\end{equation}
with $h_a(\mathcal{C})=\sum_{f\in \mathcal{C}} q_a^f$.

Let $\mathcal{C}^*(\mathcal{S})$ denote the set of contents that maximize savings \eqref{eq:savings} and write $\lambda_a^f$ for $T_aq_a^f$, the residual demand routed to the CO for content $f$ from ANO $a$ users. The following proposition characterizes $\mathcal{C}^*$.
 \begin{prop} \label{prop:percontent}
 \begin{equation} \label{eq:optset}
 \mathcal C^*(\mathcal{S}) = \Big\{f \in {\mathcal F} : \sum_{a \in \mathcal{S}}{\lambda_a^f} > s / b \Big\}.
 \end{equation}
\end{prop} 
\begin{proof}
We can re-write \eqref{eq:savings} as
\begin{equation}\label{eq:CF:3}
E(\mathcal{C},\mathcal{S})= \sum _{f \in \mathcal C} \Big(\sum_{a \in \mathcal{S}} {(\lambda_a^f}\cdot b) - {s} \Big).
\end{equation}
Clearly, $E(\mathcal{C},\mathcal{S})$ is maximized by including in $\mathcal{C}$ all contents for which the term in parentheses is positive. 
\end{proof}

\subsection{Game between CP and ANO}\label{sec:1ANO}

To better understand the relation between CP and ANOs, we first assume the CP cache is dedicated to a single ANO, denoted $\ANO$, and consider the game $G = (\{\CP, \ANO\}, v)$ with value function
$$v(\{\CP, \ANO\}) = E(\mathcal{C}^*,\ANO), $$
where $\mathcal{C}^* =\argmax_{\mathcal{C} \subseteq \mathcal{F}} E(\mathcal{C},\ANO)$, and 
 $v(\mathcal{S})=0\textrm{ for any }\mathcal{S} \subset \{\CP, \ANO\}.$
 It is straightforward to verify that the core of this game consists of shares 
\begin{align*}
\Phi_{\CP} & = r \times E(\mathcal{C^*}, \ANO) \\ %(T_a h_a(\mathcal{C^*}) b - C^* s)\\
\Phi_{\ANO} & = (1-r) \times E(\mathcal{C^*}, \ANO), %(T_a h_a(\mathcal{C}^*) b - C^* s),
\end{align*}
 for $0 \le r \le 1$.  

Any member of the core with $r>0$ has the essential property that $\CP$ will self-interestedly choose to cache the maximizing set of contents $\mathcal{C}^*$. Moreover, any $r\in (0,1)$ brings positive gain to both players and might be considered to define a suitable outcome. 

In considering a similar tradeoff game, Douros \etal  \cite{douros2017caching} propose to use the Shapley value which, in the present case can readily be shown to correspond to equal shares, $r=0.5$. It is not obvious, however, that such a choice would be acceptable to either $\CP$ or $\ANO$. We believe the appropriate value should reflect the relative bargaining power of $\CP$ and $\ANO$ (depending, for instance, on the size of their respective customer bases) and be decided by bilateral negotiation. Such negotiation determines a \emph{weighted} Shapley value, as defined in \cite{kalai1987weighted}.  
%\blue{by bilateral compensations between \CP and \ANO}

We suppose the optimal tradeoff will be realized as follows. $\CP$ makes the cache placement $\mathcal{C^*}$ that maximizes its subsidy $\Phi_{\CP}$ based on its private estimate of  $\{\lambda_{\tiny{ \ANO}}^f\}$. Learning the placement size $C^*$, $\ANO$ pays the CO $C^*s$ for storage and pays the transit ISP $T(1-h(\mathcal C^*))b$ for traffic. Note that the latter payment would be for realized traffic that might differ from the CP estimate. $\ANO$ additionally makes a side payment of $r(Th(\mathcal C^*)b - C^*s)$ to $\CP$. The correctness of payments is verifiable since both $\CP$ and $\ANO$ are aware of cache size $C$ and are able to measure traffics $T$ and $Th(\mathcal{C}^*)$. %$\CP$ has an incentive to correctly estimate demand \red{(obvious?)}.

\subsection{Game between multiple ANOs}\label{sec:BM:multi-MNO}

Consider now multiple ANOs that can all download content from the same CP cache. The objective is to determine how the ANOs should share the cost of the cache and realize appropriate side payments to $\CP$.  

%\red{Define the game $G(\mathcal{A},v)$ where the value function is $v(\mathcal{S})}
Define the game $G(\mathcal{A},v)$ where the value function is $v(\mathcal{S}) = E(\mathcal{C}^*, \mathcal{S})$, given by  \eqref{eq:savings} with $\mathcal{C}^*$ the maximizing set of contents for $\mathcal{S} \subseteq \mathcal{A}$. This game does not explicitly include $\CP$ since we assume its share is determined by the share parameter $r_a$ negotiated independently with each ANO, as discussed in Sec. \ref{sec:1ANO}.
The value distribution over $\mathcal S$, denoted $\{\Phi_a(\mathcal{C}, \mathcal{S})\}$ defines the saving of ANO $a$ for $a\in \mathcal{S}$,  to be shared with $\CP$ so that the overall subsidy is 
\begin{equation} \label{eq:subsidy}
\textsf{sub}(\mathcal{C},\mathcal{S}) = \sum_{a\in \mathcal{S}} r_a \Phi_a(\mathcal{C},\mathcal{S}).
\end{equation}

We claim the distribution $\{\Phi_a(\mathcal{C}, \mathcal{A})\}$ should ideally satisfy the following properties:
\begin{itemize}
\item \emph{coalition incentive}: the distribution should encourage all ANOs to share the same $\CP$ cache; this will happen when the distribution $\{\Phi_a(\mathcal{C}, \mathcal{A})\}$ is in the core of game $G$,  
%since overall cost decreases as more ANOs share the same CP cache, the distribution $\{\Phi_a(\mathcal{C}, \mathcal{A})\}$ should be in the core of game $G$,  
\item \emph{efficiency}: total distributed value should be equal to the value function, \ie, $\sum_{a\in \mathcal{A}} \Phi_a(\mathcal C, \mathcal A) = E(\mathcal{C},\mathcal{A})$,
\item \emph{optimality}: the $\CP$ subsidy \eqref{eq:subsidy} must be proportional to $E(\mathcal{C},\mathcal{S})$ given by \eqref{eq:savings} to induce it to realize the overall optimal placement $\mathcal{C}^*$,
\item \emph{neutrality}: the $\CP$ placement $\mathcal{C}^*$ should be independent of the respective bargaining power of ANOs as manifested by the $r_a$,
\item \emph{verifiability}: the savings distribution should be readily computable and verifiable by the ANOs and $\CP$.
\end{itemize}
Supposing ANOs should only pay for their own traffic while sharing the cost of storage, we consider distributions of the form,
 \begin{eqnarray}
\Phi_a (\mathcal C, \mathcal{S}) &=&  T_a h_a(\mathcal C) \cdot b -   {\zeta_a (\mathcal{C})  C \cdot s}, \label{eq:distribution1} \\ 
&=& \sum_{f \in \mathcal{C}} \left(\lambda_a^f \cdot b - \eta_a (f) \cdot s \right),\label{eq:distribution2}
 \end{eqnarray}
where $\zeta_a(\mathcal{C})$ and $\eta_a(f)$ are fractional shares, to be defined, and such that $\sum_{f \in \mathcal{C}} \eta_a(f) =\zeta_a(\mathcal{C}) C$ and $\sum_{a \in \mathcal{S}} \zeta_a(\mathcal{C}) = 1$. The latter condition is required for efficiency.

Theorem \ref{th:eta} defines the unique distribution of this form that is valid for general traffic and arbitrary subsidy fractions. %This distribution is not tractable however except in the special cases identified in Proposition \ref{prop:tractable}.

\begin{theorem} \label{th:eta}
Under general traffic $\{\lambda_a^f\}$ and   arbitrary subsidy  fractions $\{r_a\}$,
the distribution $\{\Phi_a(\mathcal C, \mathcal S)\}$ in \eqref{eq:distribution2} is optimal and neutral over coalition $\mathcal S \subseteq \mathcal A$, if and only if, 
\begin{equation} \label{eq:eta}
\eta_a (f) = \frac{\lambda_a^f} {\sum_{n \in \mathcal{S}} \lambda_n^f}, \forall a \in \mathcal S. 
\end{equation}
Moreover, the distribution $\Phi_a(\mathcal C, \mathcal A)$ is in the core of $G(\mathcal{A},v)$.
\end{theorem}

\begin{proof}
From \eqref{eq:subsidy} and \eqref{eq:distribution2}, the $\CP$ subsidy is
\begin{equation}\label{eq:sub:eta}
 \textsf{sub}(\mathcal{C},\mathcal{S}) = \sum_{f \in \mathcal{C}} \bigg( \sum_{n \in \mathcal{S}} r_n \lambda_n^f\cdot b - \sum_{n \in \mathcal{S}} r_n \eta_n(f) \cdot s \bigg).
 \end{equation}
This is maximal if and only if $\CP$ chooses to cache contents $f$ such that
\begin{equation}
\frac{ \sum_{n \in \mathcal{S}} r_n \lambda_n^f} {\sum_{n \in \mathcal{S}} r_n \eta_n(f)} > \frac{s}{b}. \label{eq:condition}
\end{equation}
For optimality, this condition must coincide with \eqref{eq:optset} so that,
$$ \sum_{n \in \mathcal{S}} \lambda_n^f = \frac{ \sum_{n \in \mathcal{S}} r_n \lambda_n^f} {\sum_{n \in \mathcal{S}} r_n \eta_n(f)}, \forall f \in \mathcal F .$$
For neutrality, this equation must be an identity with respect to the $\{r_a\}$ yielding expression \eqref{eq:eta} on setting particular values $r_a=1$ and $r_n=0$ for $n \neq a$. %\red{Since the proposed values are $r \in (0, 1)$, isn't it better to write $r_a=r$ and $\mathcal S=\{a\}$?} JR: seems more elegant as is and clearly true for $r_a \rightarrow 1$ and $r_n \rightarrow 0$ if we really need $r \in (0,1)$.
%particular values $r_a=r$ and $r_n=0$ for $n \neq a$.

To prove sufficiency, using \eqref{eq:eta} in \eqref{eq:condition} gives the condition of Proposition \ref{prop:percontent}. The shares are thus optimal. Sharing is also neutral since the set defined by Proposition \ref{prop:percontent} is independent of $\{r_a\}$.%since the $\eta_a(f)$ are independent of the $r_a$ \blue{since the set defined by Proposition \ref{prop:percontent} is independent of $r_a$}.
 
% \begin{prop}\label{prop:core}
%The distribution defined in Proposition \ref{prop:eta} is in the core of game $G(\mathcal{A},v)$. 
%\end{prop}
%\begin{proof}

To prove $\{\Phi_a(\mathcal C, \mathcal A)\}$ is in the core, we deduce from optimality that
\begin{align*}
 v(\mathcal S) &= E(\mathcal{C}^*,\mathcal{S}) = \sum_{f \in \mathcal C^*(\mathcal S)} \Big(\sum_{a \in \mathcal S}( \lambda_a^f \cdot b) -  s\Big) \\
 & = \sum_{f \in {\mathcal C}^*(\mathcal S)} \sum_{a \in \mathcal S} \Big(\lambda_a^f \cdot b - \frac{\lambda_a^f}{\sum_{n \in \mathcal S} \lambda_n^f} \cdot s \Big) \\
 & \le \sum_{f \in \mathcal C^*(\mathcal A)} \sum_{a \in \mathcal S} \Big(\lambda_a^f \cdot b - \frac{\lambda_a^f}{\sum_{n \in \mathcal A} \lambda_n^f} \cdot s\Big) \\
 &=  \sum_{a \in S}{\Phi_a(\mathcal C^*(\mathcal A), \mathcal A)}, 
 %v(\mathcal A),
\end{align*}
where the last inequality follows from the fact that any content $f$ included in the optimal set $\mathcal C ^*(\mathcal S)$ is necessarily also included in $\mathcal C^* (\mathcal A)$. Thus no sub-coalition has value greater than that of the grand coalition.
\end{proof}

The distribution defined in Theorem \ref{th:eta} is unfortunately not verifiable by the ANOs since $\zeta_a(\mathcal{C})$ depends on the demand distributions $\lambda_a^f$ that are unknown to them. Proposition \ref{prop:zeta} defines a verifiable distribution that is optimal in some special cases.
 
 \begin{prop} \label{prop:zeta}
Over coalition $\mathcal S \subseteq \mathcal A$, the distribution $\{\Phi_a(\mathcal C, \mathcal S)\}$ defined in \eqref{eq:distribution1} is optimal and neutral if 
 \begin{equation}
\zeta_a (\mathcal C) = \frac{T_a h_a (\mathcal{C})} {\sum_{n \in \mathcal{S}} T_n h_n (\mathcal{C})},  \label{eq:zeta}
\end{equation}
and all ANOs have the same popularity distribution, $q_a^f= q^f$ for $a \in \mathcal{S}$, and optimal if \eqref{eq:zeta} holds and all ANOs apply the same subsidy fraction, $r_a=r$ for $a \in \mathcal{S}$.
 \end{prop}
 \begin{proof}
With $\Phi_a$ given by \eqref{eq:distribution1} and \eqref{eq:zeta}, the $\CP$ subsidy may be written,
\ifCLASSOPTIONonecolumn 
\begin{equation}\label{eq:sub:zeta}
\textsf{sub}(\mathcal{C},\mathcal{S}) = \frac{\sum_{n \in \mathcal{S}} r_n T_n h_n(\mathcal{C})} {\sum_{n \in \mathcal{S}}  T_n h_n(\mathcal{C})} \times 
\bigg(\sum_{a \in \mathcal{S}} T_a h_a(\mathcal{C})\cdot b - C \cdot s \bigg). 
\end{equation}
\else
\begin{multline}\label{eq:sub:zeta}
\textsf{sub}(\mathcal{C},\mathcal{S}) = \frac{\sum_{n \in \mathcal{S}} r_n T_n h_n(\mathcal{C})} {\sum_{n \in \mathcal{S}}  T_n h_n(\mathcal{C})} \times \\
\bigg(\sum_{a \in \mathcal{S}} T_a h_a(\mathcal{C})\cdot b - C \cdot s \bigg). 
\end{multline}
\fi
The term in parenthesis is the overall profit so that $\CP$ will place the optimal set $\mathcal{C}^*$ if the pre-factor is independent of $\mathcal{C}$. This occurs if the popularity distributions are the same for all when $h_n(\mathcal{C})=h(\mathcal{C})$. In this case the distribution is optimal and neutral. The pre-factor is also constant if all the $r_a$ are equal proving optimality. 
\end{proof}
%\blue{Note: The distribution in Proposition \ref{prop:zeta} does not lead to the optimal if the CP will decide about the cache content. The distribution in Proposition \ref{prop:eta}
%is optimal and neutral, but not tractable for the ANO. One can assume that ANO can estimate the subsidy by $\zeta_a(\mathcal{C})$ while the CP will optimize according to $\eta_a(f)$. Numerical evaluations should show how subsidy will differ in these cases and how suboptimal is $\zeta_a(\mathcal{C})$.} 

%\noindent \emph{Discussion}.

The value distributions determined from Theorem \ref{th:eta} and Proposition \ref{prop:zeta} are different except in the special case where all popularity laws are the same. When the laws are different and fractions $r_{ak}$ are not the same, the value distribution determined from \eqref{eq:zeta} is not optimal. Moreover, the CP is required to maximize the subsidy defined by \eqref{eq:sub:zeta} which is non-trivial, especially in the context of the access network as considered in the next section. The optimal distribution determined from \eqref{eq:eta} is therefore preferable despite the impossibility for ANOs to independently verify their allocation. 
%The distribution of Theorem \ref{th:eta} can be computed by $\CP$ but not by the ANOs who do not know, and cannot measure, the individual popularity distributions $q_a^f$. They might trust $\CP$ to make the placement that maximizes its subsidy but cannot independently verify this. The problem does not arise if all ANOs have the same popularity distributions or if a regulator imposes a given share $r$, when the simpler distribution given by Proposition \ref{prop:zeta} applies. This distribution is verifiable assuming the ANOs can learn the overall volume of traffic served by the cache and can measure their own share of this.

We envisaged applying a distribution based on the Shapley value. One possibility is to derive the Shapley value for a fixed cache size $C$, compute the $\CP$ subsidy based on this, and then suppose $\CP$ places the subsidy maximizing content $\mathcal{C}$. This distribution does not have the optimality property, however: the $\CP$ does not have the correct incentive to optimize the overall memory for bandwidth tradeoff. To compute the Shapley value on supposing content  placement is optimized for each ordered sub-coalition, on the other hand, rapidly becomes computationally intractable (as in \cite{ma2010cooperative}, for instance).

\subsection{Approximate verification}

The value distribution determined from Theorem \ref{th:eta} has all the desirable properties except verifiability. In this section, we numerically explore the possibility of using $\zeta_a$ from \eqref{eq:zeta} in \eqref{eq:distribution1} to perform an approximate verification. In other words, $\CP$ computes the optimal set $\mathcal{C}$ and distribution $\{\Phi_a(\mathcal C, \mathcal A)\}$ by applying Theorem \ref{th:eta} while each ANO estimates the realized shares using \eqref{eq:zeta} and \eqref{eq:distribution1}. There is of course, a difference between  predicted demand (the $\lambda^f_a$) and realized demand (measured $T_ah_a(\mathcal C)$) but we ignore this discrepancy here and consider the error arising when realized demand actually coincides with the forecast.

We consider a CP with  a catalogue of $F=10^7$ megabyte files. The fixed cost of storage at the CO is $s=\$0.03$ per GB (\$$3 \times 10^{-5}$ per content) and the cost of bandwidth is $b=\$ 4$ per Mb/s. Two ANOs share the CP cache and have demand $T_1=160$ Mb/s and $T_2 = 80$ Mb/s. The popularity law for both is Zipf(0.8) but the ranking of contents for ANO $2$ is a random permutation of the ranking for ANO $1$. Numerical results are the average of 10 different random permutations. 

   \begin{figure}[t]
         \centering
         \includegraphics[width=\ifCLASSOPTIONtwocolumn 0.8\columnwidth \else 0.5\columnwidth \fi]{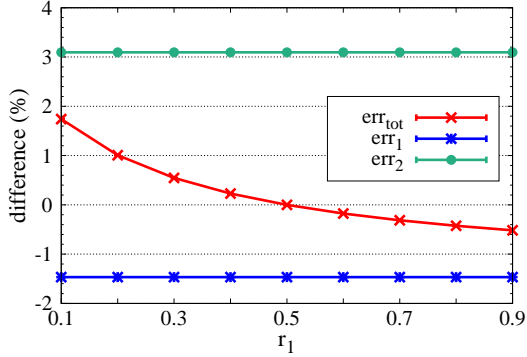}
         \caption{Percentage difference when verifying subsidies computed using Theorem \ref{th:eta} by applying formula \eqref{eq:zeta}; \ifCLASSOPTIONtwocolumn \\ \fi 2 ANOs, $T_1=160$ Mb/s, $T_2=80$ Mb/s, $F=10^7$, uncorrelated Zipf($0.8$) popularities, $r_2=0.5$.}          \label{fig:BM:errors}
            \end{figure}

Fig. \ref{fig:BM:errors} plots the percentage difference between the required subsidy using \eqref{eq:zeta}, $\textsf{sub}_a(\zeta)$, and the subsidy due using \eqref{eq:eta}, $\textsf{sub}_a(\eta)$, as a function of $r_1$ when $r_2=0.5$ is fixed. More precisely, we plot 
\begin{equation*}
\textsf{err}_a = \left(\textsf{sub}_a(\zeta) - \textsf{sub}_a(\eta)\right)/ \textsf{sub}_a(\eta)\times 100
\end{equation*}
against $r_1$, for $a=1$ and $a=2$, and the equivalent relative error, $\textsf{err}_{\textrm{tot}}$, computed for the total subsidies \eqref{eq:sub:zeta} and \eqref{eq:sub:eta}. The $95\%$ confidence interval estimated from the $10$ different ranking permutations is  smaller than the line thickness.
%The figure shows 95\% confidence intervals estimated from the $10$ different ranking permutations.

As predicted by Prop. \ref{prop:zeta}, the total subsidy is accurately estimated when $r_1=r_2$. The errors for individual subsidies, however, do not vary with $r_1$ since both the true shares computed via $\eta_a(f)$ and the estimates computed via $\zeta_a(\mathcal C$) do not depend on the $r_a$. Importantly, the error is small despite the extreme disparity between popularity laws. The absolute error  increases with the ratio $T_1/T_2$ but remains relatively small (\eg, $\textsf{err}_{\textrm{2}} < 15\%$ for $T_1/T_2 < 10$). These results suggest it may be possible to design a practical subsidy scheme that combines the optimality of \eqref{eq:distribution2} and \eqref{eq:eta} with the verifiability of \eqref{eq:distribution1} and \eqref{eq:zeta}. Such a design is beyond present scope, however.

\section{Optimizing the tradeoff}\label{sec:optimization}

We first formulate the general tradeoff optimization problem under capacity constraints before elaborating a distributed solution  based on Lagrangian relaxation for the particular network instance of Fig. \ref{fig:network}. Numerical results demonstrate how the proposed method works on a toy network example. %based on synthetic data are provided to elaborate the performance of our method.
\subsection{General problem formulation}\label{sec:general}
We assume content placement is optimized on a daily basis. CPs have detailed demand forecasts while ANOs are aware of capacity limits. Externally fixed costs are known to both CPs and ANOs. Optimal placements are realized at the start of the day in an off-peak traffic period.  Charges and subsidies are calculated at the end of the day based on  allocated cache capacities,  measured peak period demand and  respective unit resource prices.

We use the notation introduced in Sec. \ref{sec:ufl} with the addition of storage and bandwidth capacity limits for each node $n \in \mathcal{N}$, denoted $S_n$ and $B_n$, respectively. We effectively assimilate the path from $n$ to its parent to a link of capacity $B_n$. CPs,   $k \in \mathcal{K}$, and ANOs, $a \in \mathcal{A}$, are implicitly identified by disjoint subsets of contents, $f \in \mathcal{F}_{k}$, and leaf nodes, $l \in \mathcal{L}_{a}$, respectively. All intermediate nodes, $\mathcal{I} = \mathcal{N} \setminus \{\mathcal{L},0\}$, also belong exclusively to one ANO. Let $\mathcal{I}_{a}$ be the set of intermediate nodes belonging to ANO $a$ and denote all nodes belonging to $a$ by $\mathcal{N}_a = \mathcal{L}_a \cup \mathcal{I}_a$.  

%where $n \prec m$ means node $m$ precedes $n$ on the path from a leaf to the source.
We also need notation for the cost of a path from each node $n$ to the source: $b_{n+} = \sum_{n \preceq m \preceq 0} b_m$ where $n \preceq m \preceq 0$ covers the set of links on the path from $n$ to the source. Costs are interpreted here as usage based {\em{prices}}. Each ANO can fix its own prices but, for the sake of neutrality, these are applied equally to all CPs.

Content placement is defined by indicator variables: $x_n^f=1$ if content $f$ is stored in node $n$ and $x_n^f=0$ otherwise. We additionally define delivery variables: $y_{ln}^f=1$ if requests from leaf $l$ for content $f$ are delivered from node $n$ and  $y_{ln}^f=0$ otherwise. In the optimal answer, we want to have $y_{ln}^f=1$ if $n$ is the first node such that $x_n^f=1$ on the path from $l$ to the source.  The utility of a placement is equal to the overall savings relative to a network without caches. We seek to solve the following problem. %\blue{We seek to derive the optimal placement such that the utility of the system consisting of all all ANOs and all CPs become maximized. }
\ifCLASSOPTIONonecolumn
\begin{flalign}
\textbf{(ILP)} \hspace{0.25 \columnwidth} \textrm{maximize } &    \,U({x}, {y})  = \sum_{f \in \mathcal{F}}  U({x^f}, {y^f}), \textrm{   where      } &\notag \\
 U({x^f}, {y^f}) &=  \sum_{l \in \mathcal L} \sum_{l \preceq n \preceq 0} { \lambda_l^f y_{ln}^f} b_{n+}  \,-\,  \sum_{n \in {\mathcal N}} x_n^f  s_{n} , & \label{eq:utility} 
 \end{flalign}
\else
\begin{flalign}
\textbf{(ILP)} \hspace{0.05 \columnwidth} \textrm{maximize } &    \,U({x}, {y})  = \sum_{f \in \mathcal{F}}  U({x^f}, {y^f}), \textrm{   where      } &\notag \\
 U({x^f}, {y^f}) &=  \sum_{l \in \mathcal L} \sum_{l \preceq n \preceq 0} { \lambda_l^f y_{ln}^f} b_{n+}  \,-\,  \sum_{n \in {\mathcal N}} x_n^f  s_{n} ,  \label{eq:utility} &
 \end{flalign}
\fi
subject to placement and delivery constraints for each content $f \in \mathcal F$,
\begin{align}
 \sum_{n \in \mathcal N} y_{ln}^f \leq 1,&\,\,\forall l,\label{eq:UFL:ILP:rout}\\
 y_{ln}^f \leq x_{n}^f,\,\forall l,&\,\,\forall n, \label{eq:UFL:ILP:place}\\
x_{n}^f \in \{0, 1\},&\, \,\forall n,  \label{eq:UFL:ILP:x} \\
y_{ln}^f \in \{0, 1\},& \,\,\forall l,n, \label{eq:UFL:ILP:y}
\end{align}
and capacity constraints,
\begin{align}
& \sum_f x_n^f \le S_n,\,\,\forall n, \label{eq:storagecap}\\
& \sum_f \lambda^f_l (1- \sum_{l \preceq m \preceq n} y^f_{lm}) \le B_n,\,\,\forall n. \label{eq:bandwidthcap}
\end{align}
%This integer linear program that is a generalization of the classical, NP-hard, capacitated facility location problem. (where only bandwidth constraints \eqref{eq:bandwidthcap} would be imposed) 
% In the capacitated facility location, the capacity is consumed only for the requests which are assigned to that node. In our case, we have a hierarchy. Thus, it is still a generalization of CFL when only constraint \eqref{eq:bandwidthcap} applies.}
This integer linear program is a generalization of the NP-hard problem of optimal content placement in cache hierarchies \cite{poularakis2016complexity} (where only storage constraints \eqref{eq:storagecap} would be imposed).
Before presenting a distributed solution based on Lagrangian relaxation, we discuss the constraints that apply in our particular instance. 
\subsection{A specific instance}
We consider the three-tier network of Fig. \ref{fig:network}. The root node 0 is the CO and, as assumed in Sec. \ref{sec:shareCO}, storage and bandwidth are unlimited with externally fixed prices $s_0$ and $b_0$, respectively. Intermediate nodes, $i \in \mathcal{I}$, are also assumed able to elastically provide as much storage as needed at fixed unit price $s_i$. On the other hand, we suppose storage in leaf nodes is limited to $S_l < \infty$, for $l \in \mathcal{L}$, typically for technological reasons.

Bandwidth in the access network is supposed to be provisioned based on demand forecasts with a lead time of several months. This means bandwidth is limited but typically more than might be optimal at any considered instant within the planning cycle since demand has an increasing trend. We assume network planners are competent and ensure that no bandwidth constraint is violated when available cache space is optimally used, \ie, the problem does have a feasible solution. % with no demand overload. 

For capacity limited resources, the ANOs fix prices that induce CPs to optimally share available capacity. One objective would be, for instance, to fully use the limited cache in a wireless base station. Similarly, the bandwidth price on the link from an intermediate node $i$ to the CO can be set to ensure maximum utilization and thus minimize the cost of storage at that node. Storage and bandwidth are effectively \emph{substitutable} resources: the relative elasticity of storage provision at the intermediate nodes enables more efficient use of provisioned bandwidth. 

\subsection{A distributed solution}
The complexity of\textbf{ (ILP)}  and the separate content placement and price setting roles of the CPs and ANOs, respectively, impose a distributed solution. This is possible through dual Lagrangian decomposition by relaxing constraints \eqref{eq:storagecap}  and \eqref{eq:bandwidthcap} and applying the sub-gradient method to derive an iterative procedure. The relaxed problem is then a set of UFL instances as introduced in Sec. \ref{sec:ufl}  and can be solved independently by the CPs. %Note that similar relaxation have been considered in \cite{klincewicz1986lagrangian} for capacitated facility location problem but the constraint has a different structure and the optimization is done in a non-distributed way by one entity.  JR: I don't think this is useful

First introduce notation for the cache size used by CP $k$ at node $n$,
\begin{equation}
C_{nk} =  \sum_{f \in \mathcal F_{k}} x_n^{f}, \label{eq:cache} 
\end{equation}
and the residual amount of CP $k$ demand routed over link $n$,
\begin{equation}
\Lambda^k_{n} = \sum_{l \in \mathcal{L}} \sum_{f \in \mathcal F_{k}} \lambda^f_{l} (1- \sum_{l \preceq m \preceq n} y_{lm}^{f}). \label{eq:Lambda}
\end{equation}
For link $0$, from CO to the source, we need the per-ANO partition of $\Lambda_{0}^k$. Let
\begin{equation}
\Lambda_{0a}^k = \sum_{l \in \mathcal{L}_a} \sum_{f \in \mathcal F_{k}} \lambda^f_{l} (1- \sum_{l \preceq m \preceq 0} y_{lm}^{f}). \label{eq:Lambda0a}
\end{equation}
Finally, let $\zeta_{ak}$ be the ANO $a$ share of the CO cache cost of CP $k$, computed from \eqref{eq:distribution2} and \eqref{eq:eta},
\begin{equation}
\zeta_{ak} = \frac{1}{C_{0k}} \sum_{f \in \mathcal{C}_{0k}} \frac{\sum_{l \in \mathcal{L}_a} \lambda_l^f y_{l0}^f } { \sum_{l \in \mathcal{L}} \lambda_l^f y_{l0}^f }. \label{eq:zetaak}
\end{equation}

For our specific network instance, utility can then be expressed as,
\begin{align}\label{eq:utility:total}
U(x,y) &= \sum_{a \in \mathcal{A}} U_{a}(x,y),
\end{align}
where 
\ifCLASSOPTIONonecolumn 
\begin{equation} \label{eq:utilitya}
U_{a}(x,y)  =   \sum_{k \in \mathcal{K}} \Big((\sum_{l \in \mathcal L_a} T^k_{l} -  \Lambda^k_{0a}) b_0 - \zeta_{ak} C_{0k} s_0 - \sum_{i \in \mathcal{I}_a }  C_{ik} s_i \Big),
\end{equation}
\else
\begin{multline} \label{eq:utilitya}
U_{a}(x,y)  =   \sum_{k \in \mathcal{K}} \Big((\sum_{l \in \mathcal L_a} T^k_{l} -  \Lambda^k_{0a}) b_0\, -\\
- \zeta_{ak} C_{0k} s_0 - \sum_{i \in \mathcal{I}_a }  C_{ik} s_i \Big),
\end{multline}
\fi
and $T_{l}^k= \sum_{f \in \mathcal{F}_k} \lambda_l^f$ is the demand for CP $k$ at leaf $l$. The only costs appearing in \eqref{eq:utility:total} are $b_0$, $s_0$ and $s_i$ for $i \in \mathcal{I}$. Capacity limited resources are considered as sunk costs and the corresponding $b_n$ and $s_n$ are set to zero. Their utilization is controlled by shadow prices in the form of Lagrange multipliers.

Introduce the Lagrangian multipliers $\sigma_l \ge 0$ and $\beta_n \ge 0$ for the relaxed constraints \eqref{eq:storagecap} and \eqref{eq:bandwidthcap}, respectively, and define the Lagrangian,
$$L(x,y,{\beta},{\sigma}) =  \sum_{a \in \mathcal{A}} L_a(x,y,{\beta},{\sigma}), $$
where
\ifCLASSOPTIONonecolumn
\begin{equation}\label{eq:lagrangea}
L_a(x,y,{\beta},{\sigma}) =  U_a(x,y) -   \sum_{l \in \mathcal{L}_a} \sigma_l \Big(\sum_{k \in \mathcal{K}} C_{lk} - S_l \Big) - \sum_{n \in \mathcal{N}_a} \beta_n \Big( \sum_{k \in \mathcal{K}} \Lambda^k_{n} - B_n \Big) ,
\end{equation}
\else
\begin{multline}\label{eq:lagrangea}
L_a(x,y,{\beta},{\sigma}) =  U_a(x,y) -   \sum_{l \in \mathcal{L}_a} \sigma_l \Big(\sum_{k \in \mathcal{K}} C_{lk} - S_l \Big) \,- \\
 - \sum_{n \in \mathcal{N}_a} \beta_n \Big( \sum_{k \in \mathcal{K}} \Lambda^k_{n} - B_n \Big) ,
\end{multline}
\fi
and $U_a(x,y)$ is given by \eqref{eq:utilitya}. The dual of (\textbf{ILP}) can then be written,
%and $U_a(x,y)$ is given by \eqref{eq:utilitya}, $x$ and $y$ are subject to constraints \eqref{eq:UFL:ILP:rout} to \eqref{eq:UFL:ILP:y} and the $\beta_n \ge 0$ and $\sigma_l \ge 0$ are Lagrange multipliers. The dual of (\textbf{ILP}) can then be written,
\ifCLASSOPTIONonecolumn
\begin{flalign}\label{eq:opt:dual}
\textbf{(DP)}  \hspace{0.3\columnwidth} & \min_{\beta \ge 0,\, \sigma \ge 0}  \max_{x,\,y}\, L(x,y,{\beta},{\sigma}),& 
\end{flalign}
\else
\begin{flalign}\label{eq:opt:dual}
\textbf{(DP)}  \hspace{0.2\columnwidth} & \min_{\beta \ge 0,\, \sigma \ge 0}  \max_{x,y}\, L(x,y,{\beta},{\sigma}),& 
\end{flalign}
\fi
where $x$ and $y$ are subject to constraints \eqref{eq:UFL:ILP:rout} to \eqref{eq:UFL:ILP:y}. We propose a distributed solution for (\textbf{DP}) where each CP places content to maximize $L$ for given values of ${\beta}$ and ${\sigma}$ while the ANOs iteratively adjust these values using the sub-gradient method \cite{fisher1981lagrangian}.
More precisely, we suppose the CPs and ANOs perform the distributed optimization at the start of each day, exchanging information in  rounds of primal and dual update cycles. A third party `orchestrator' (an algorithm executed in the CO, say)  is additionally informed of iteration outcomes to determine step sizes and stopping conditions. The respective operations are summarized in Algorithm \ref{alg:opt} and detailed below.

\begin{algorithm}[t]
\centering
\caption{Online algorithm to optimize utility \eqref{eq:utility:total}. O denotes the orchestrator.} \label{alg:opt}.   
\begin{algorithmic}[1]
    \Ensure  Optimal  $x^f_n$ and $y^f_{ln}$, \,$\forall f, n, l$
%    \Require O: sets $\tau=1$, $\textsf{LB}=-\infty$, $\textsf{UB}=\infty$, $\epsilon$ and $\gamma$. ANO $a$: chooses initial $\sigma_l(0) \ge 0, \,\forall l$, $\beta_n(0) \ge 0, \,\forall n \in \match N_a$.
    \Require {  \quad \newline O: sets $\tau=1$, $\textsf{LB}=0$, $\textsf{UB}=\sum_{k} \sum_{l} T_l^k b_0$, $\epsilon$ and $\gamma$. \newline ANO $a$: chooses initial $\sigma_l(0) \ge 0$, $\beta_n(0) \ge 0$.}

    \While {$\textsf{UB} - \textsf{LB} > \epsilon$ and $\tau < \tau_{\text{max}}$ }
    \State CP $k$: calculates optimal $x^f(\tau)$ and $y^f(\tau), \forall f \in \mathcal F_k$ 
   \State \qquad \,\, communicates the results \eqref{eq:cache}, \eqref{eq:Lambda}, \eqref{eq:Lambda0a} 
    \NoNumber{\qquad \,\,  and \eqref{eq:zetaak} to each ANO $a$ and to O.}
      %\NoNumber{ communicates the result to ANO $a$ in form of $C_{nk}(\tau)$}
      %\NoNumber{ $\Lambda_{nk}(\tau),\forall n \in \mathcal N_a$, $\Lambda^a_{0k}(\tau)$ and  $\zeta_{ak}(\tau)$.}      
     \State O: $\textsf{LB} \leftarrow \max \{\textsf{LB}, U(x(\tau), y(\tau))\}$.
    \State \quad\, $\textsf{UB} \leftarrow  \min \{\textsf{UB}, L(x(\tau), y(\tau), \beta(\tau), \sigma(\tau))\}$.
     \State \quad\, updates step size using \eqref{eq:step} and sends to ANOs. 
    \State ANO $a$: updates shadow prices using \eqref{eq:opt:dual:iteration}.
    \State $\tau = \tau+1$
    \EndWhile
  \end{algorithmic}
\end{algorithm}

\subsubsection{Primal updates}
The CPs have private demand estimates $\lambda_l^f$, know fixed prices $b$ and $s$ and receive proposed shadow prices $\beta$ and $\sigma$ from the respective ANOs. Using these data, CP $k \in \mathcal{K}$ computes the optimal placement of each content $f \in \mathcal{F}_k$ by solving the UFL problem of Sec. \ref{sec:ufl}. 

After performing its part of iteration $\tau$, $\tau \ge1$, each CP $k$, communicates the UFL results to each ANO $a$ and the orchestrator in the form of cache capacities $C_{nk}(\tau)$ and residual traffic demands $\Lambda_{n}^k(\tau)$ for $n\in \mathcal{N}_a$, CO cache capacity $C_{0k}(\tau)$ and ANO cost share $\zeta_{ak}(\tau)$, and residual transit traffic $\Lambda_{0a}^k(\tau)$. 

\subsubsection{Orchestration}
The orchestrator is aware of capacity constraints and  verifies the feasibility of the CP results with regard to the primal problem \textbf{(ILP)} . If feasible, these constitute a possible solution for the (\textbf{ILP}). The objective functions $L$ and $U$ are evaluated and compared to the current upper and lower bounds, $\textsf{UB}$ and $\textsf{LB}$, respectively. As the bounds are updated, the CPs and ANOs must be informed to retain their corresponding decision variables. 

The orchestrator stops iterations if the found solution has an absolute error of less than $\epsilon $ or a limit number of iterations has been attained. In either case the ANOs and CPs are informed and can implement the currently best feasible $\textsf{LB}$ solution. If not, the orchestrator calculates a step size $\delta(\tau)$ for the next iteration using the Polyak formula \cite{fisher1981lagrangian} and sends it to all ANOs,
 \begin{equation}\label{eq:step}
 \delta(\tau) = \gamma \frac{ \left|L(x(\tau),y(\tau),{\beta(\tau)},{\sigma(\tau)})- \textsf{LB} \right|}{||\nabla(\tau)||^2},
  \end{equation}
 where $\gamma \ge 0$ is a scale factor and $||\nabla(\tau)||^2 = \sum_{n \in \mathcal{N}} \nabla_n^B(\tau)^2 + \sum_{l \in \mathcal{L}}  \nabla_{l}^S{(\tau)}^2$.$ \nabla_{n}^B{(\tau)}$ and $\nabla_{n}^S{(\tau)}$ are the sub-gradients at iteration $\tau$, 
 \begin{align}
 \nabla_{n}^B{(\tau)} = \sum_{k \in \mathcal{K}} \Lambda^k_{n}{(\tau)} - B_n,\, &\textrm{ for }n \in \mathcal{N}\setminus 0, \\
 \nabla_{l}^S{(\tau)} = \sum_{k \in \mathcal{K}} C_{lk}(\tau) - S_l,\,& \textrm{ for }l \in \mathcal{L}.
 \end{align}
\subsubsection{Dual updates}
Each ANO $a$ adjusts the shadow prices $\beta_n$ and $\sigma_n$ for capacity limited resources $n \in \mathcal{N}_a$ to more closely match the capacity constraints. Specifically, ANO $a$ computes the sub-gradients  $\nabla_{n}^B{(\tau)}$ for $n \in \mathcal{N}_a$ and  $\nabla_{l}^S{(\tau)}$ for $l \in \mathcal{L}_a$. Using step size \eqref{eq:step} communicated by the orchestrator, it computes 
\begin{align}\label{eq:opt:dual:iteration}
\beta_n{(\tau+1)}  = \left[ \beta_n{(\tau)} + \delta(\tau) \nabla_{n}^B{(\tau)} \right]^+, \,&\textrm{ for } n \in \mathcal{N}_a,\notag \\
\sigma_l{(\tau+1)}  = \left[ \sigma_l{(\tau)} + \delta(\tau) \nabla_{l}^S{(\tau)} \right]^+,\,&\textrm{ for } l \in \mathcal{L}_a, 
\end{align} 
where $[x]^+$ denotes the maximum of $x$ and 0, and sends these new values to  the CPs and the orchestrator.

Convergence to the optimal dual solution (\textbf{DP}) using the Polyak step size rule has been empirically demonstrated \cite{fisher1981lagrangian} but may take many iterations for large problems. Any feasible solution after a certain number of iterations is likely to be satisfactory, however, since capacity constraints are not usually tight (e.g., bandwidth is also used for other traffic) and slight under-utilization does not have serious consequences. %To prevent the oscillation of dual variables, ANOs can stop updating dual variables which oscillate for some time until a feasible solution is found.  JR: I think this is too technical.

Due to the integer nature of decision variables, it is sometimes non-trivial to find feasible solutions. %Thus, the goal is to find Lagrange multipliers for which primal solution is feasible. 
However, it often happens that the primal solution will be nearly feasible. In these cases, the orchestrator can project the solution to a fair, sub-optimal feasible solution for any violated constraint and require the CP to update the other related cache size and traffic demand variables. Alternative more sophisticated projection heuristics, typically requiring a greater degree of cooperation between orchestrator and CPs, might also be applied  \cite{fisher1981lagrangian}.   

%In these cases, the solutions can be made feasible with a fair projection depending on the violated constraint. In the case of feasibility violation of the bandwidth capacity in the intermediate nodes or leaves, the orchestrator can retain the set of primal and dual variables as a feasible solution. If the demand also violates the cache capacity limit at some leaves, the orchestrator should project the answer to a suboptimal feasible solution for any violated constraint and ask the CP to update the other cache size and traffic demand variables. More sophisticated heuristics for projection can, of course, be designed between orchestrator and the CPs which needs more level of cooperation \cite{fisher1981lagrangian}. 
 %The standard projection mechanisms called \textit{Lagrangian heuristics} are widely discussed in the literature \cite{fisher1981lagrangian}, however our problem is different in the sense that only the CPs have the knowledge of the requests.
%The projection can be done through proportional normalization of $C_{lk}$ values with regard to the input traffic of each CP at that leaf, \ie\,$T^k_{l}$. If the updated traffic values $\Lambda^k_{n}$ constitute a nearly feasible solution, the projected point is regarded as a feasible solution. 
\subsection{Settlements}
%Let $T_{lk}= \sum_{f \in \mathcal{F}_k} \lambda_l^f$ denote demand for CP $k$ at leaf $l$. 
The computed cache capacities $C_{nk}$ are reserved for the day. Demands $T^k_{l} $ and ${\Lambda}^k_{n}$ are only estimates, however, and ANO charges and CP subsidies are based on measured busy period demand denoted $\tilde{T}^k_{l}$ and $\tilde{\Lambda}^k_{n}$. The ANOs pay for storage in intermediate nodes $i \in \mathcal{I}_a$ and the CO, and for transit bandwidth at given rates $s_i$, $s_0$ and $b_0$, respectively. They also pay a subsidy to the CPs equal to a fraction of their cost savings calculated using these fixed rates, the shadow prices and the measured traffic. ANO $a$ pays CP $k$ the following subsidy,
\begin{multline}\label{eq:sub:ak}
 \textsf{sub}_{ak} = r_{ak}  \times \Big(\sum_{l \in \mathcal{L}_a} \tilde T^k_{l} (\beta_{l}+\beta_{p(l)}+b_0 )   -  \sum_{n \in \mathcal{N}_a} (\tilde{\Lambda}^k_{n} \beta_{n})\, - \\
  - \tilde{\Lambda}_{0a}^k b_0  -  \sum_{l \in \mathcal{L}_a}( C_{lk} \sigma_l) - \sum_{i \in \mathcal{I}_a} (C_{ik} s_i) -  \zeta_{ak} C_{0k} s_0   \Big),$$
\end{multline}
where $p(l) \in \mathcal I_a$ is the parent of leaf $l$ and $r_{ak}$ is the proportionate share of overall savings negotiated between ANO $a$ and CP $k$ (see Sec. \ref{sec:shareCO}).

\subsection{Numerical application}\label{sec:optimization:evaluation}
  \begin{figure*}[t]
         \centering
         \captionsetup{justification=centering}
         \subfloat[][Shadow price $\beta_i$ vs. capacity $B_i$, \\ $S_l = 200$ GB, $B_l =$ 15 Mb/s.]{
         \includegraphics[width=0.32\textwidth]{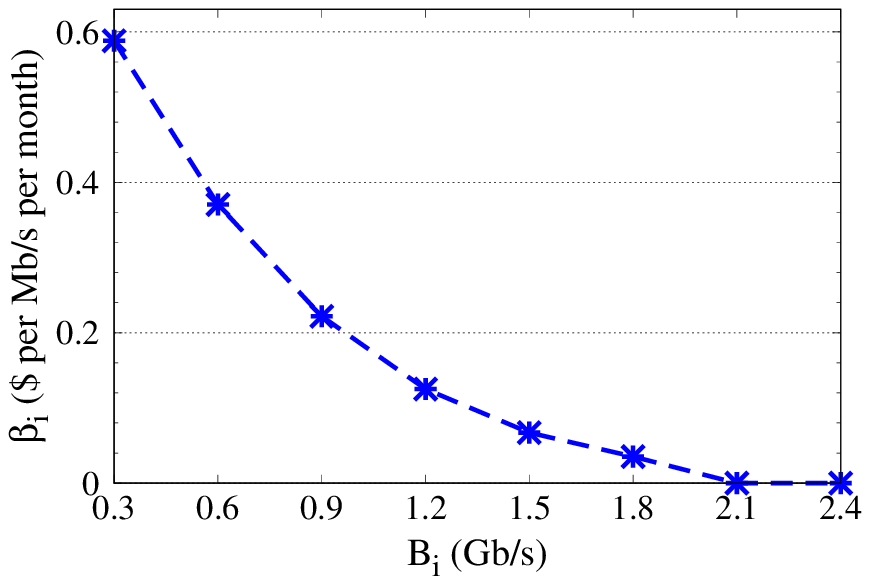}
         \label{fig:bi:price}}
          \subfloat[][Utility vs. capacity $B_i$, \\ $S_l = 200$ GB, $B_l =$ 15 Mb/s.]{
          \includegraphics[width=0.32\textwidth]{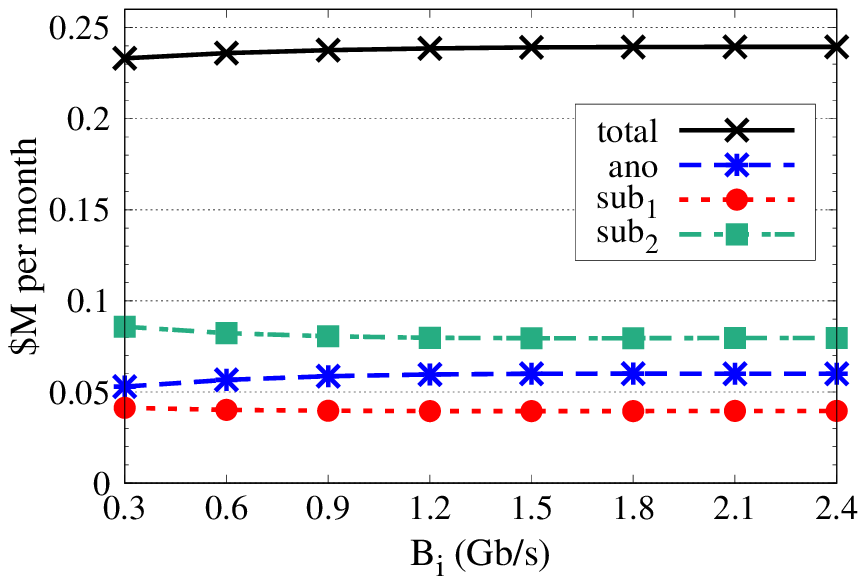}
         \label{fig:bi:gain}}
           \subfloat[][Utility vs. capacity $S_l$, \\ $B_l =$ 15 Mb/s, $B_i= 600$ Mb/s.]{
          \includegraphics[width=0.32\textwidth]{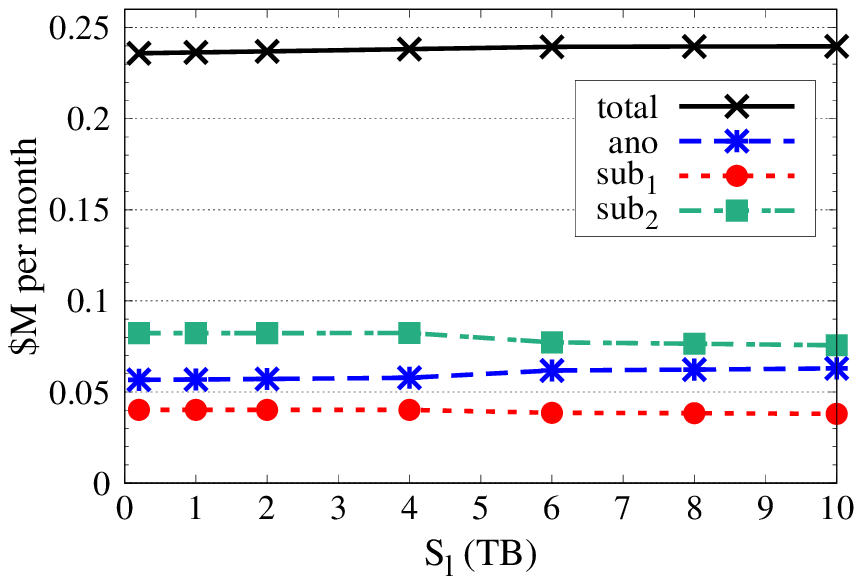}
         \label{fig:sl:gain}} 
 %       \caption{\protect \subref{fig:bi:price}, \protect \subref{fig:bi:gain} $S_l=200$ GB; \protect \subref{fig:sl:gain} $B_i = 0.6$ Gb/s: $e_1=100, e_2=10, b_0=\$4$ per Mb/s per month, $s_0=s_i=0.03$ per Gb per month.}
        \caption{Algorithm \ref{alg:opt} applied to a symmetric network with $2$ ANOs and $2$ CPs;\\ $e_1=100, e_2=10, b_0=\$4$ per Mb/s per month, $s_0=s_i=\$0.03$ per GB per month.}
        \label{fig:eval1}
    \end{figure*}
% \begin{figure*}[t]
% 	\centering
%          \subfloat[]{
%          \includegraphics[width=0.33\textwidth]{figs/plots/gradient/MBi/compareTotal/gainVSS}
%         \label{fig:bi:gain:co}}
%           \subfloat[]{
%          \includegraphics[width=0.33\textwidth]{figs/plots/gradient/MSl/compareTotal/gainVSS}
%         \label{fig:sl:gain:co}} 
%        \caption{utilities in million dollars per month when CO cache is shared between ANOs: $S_l=\rho_l(F_1+F_2), B_i= \rho_i (T^a_1+ T^a_2)/e_2, e_1=100, e_2=10, b_0=\$4, s_0=s_i=\$0.03$.}
%        \label{fig:eval1:co}
%    \end{figure*}
%per month traffic is fixed--shadow price meaning- why the plots are constant-
We consider a toy network example to illustrate an application of the method. %A thorough evaluation of its properties on a more realistic configuration is out of present scope. 
As in Sec. \ref{sec:quantify}, to simplify the presentation of results, we consider a symmetric 3-tier access network where all leaves have the same demand and all links and nodes at the same level have the same capacity. Two ANOs have dedicated leaf and intermediate node infrastructure and access the Internet via a common CO hub. Both ANOs have identical demand and network capacities. The fanouts of each ANO are $e_1=100$ and $e_2=10$, as in the example of Fig. \ref{fig:comp1}. 

ANO users access the content of two CPs both of which have a distinct content catalogue of $10^7$ one megabyte files with a Zipf(0.8) popularity distribution. Peak demand for CP $1$ content from the users of each ANO is $10$ Gb/s  while demand for CP $2$ is twice as much.  

The fixed cost of storage at the CO, $s_0$, and at the intermediate nodes, $s_i$, is $\$0.03$ per GB 
%\red{($3 \times 10^{-5}$ per content)} 
while the cost of transit bandwidth $b_0$ is set to $\$4$ per Mb/s.
%\red{($\$32$ per content download/s)}
These are monthly rates as in Sec. \ref{sec:quantify}. The other resources are provisioned by the ANOs and are considered as sunk costs. %The capacity of links from the leaves to the intermediate nodes is assumed sufficient to handle all traffic for all considered leaf cache capacities so that its shadow price is always close to zero. 
We derive optimal utilities, using Algorithm \ref{alg:opt}, as functions of the capacities of the other resources, namely leaf storage $S_l$, leaf bandwidth $B_l$ and intermediate link bandwidth $B_i$. %We have used $\gamma=0.002$ with success and the results are retuned with a maximum $\epsilon =200$ absolute error in less than $5000$ rounds.  
The negotiated fractional shares $r_{ak}$ of ANO savings used to subsidize the CPs are arbitrarily set to $0.5$ in all cases.

Fig. \ref{fig:bi:price} shows how the intermediate link shadow price $\beta_i$ varies as a function of capacity $B_i$ while leaf storage and bandwidth are fixed. %Leaf storage for this figure is fixed at $1\%$ of the total catalogue of $20$ TB giving a local hit rate of around $40\%$. 
The figure shows that  $\beta_i$ decreases as available capacity increases, going to zero when link capacity ceases to be constraining. 
Shadow prices correspond as usual to marginal utilities and reflect the quality of the ANO dimensioning. They would coincide to real costs if the network were optimally sized. %However, interpretation of network dimension is somewhat more complicated in the present case due to the substitutability of storage and bandwidth. 
%not satisfied by the leaves. 

Fig. \ref{fig:bi:gain} plots realized gains as a function of $B_i$ for fixed $S_l$ and $B_l$ and Fig. \ref{fig:sl:gain} plots the same values as a function of $S_l$ for fixed $B_l$ and $B_i$. The figures show total utility \eqref{eq:utility:total}, the subsidies received by each CP from \eqref{eq:sub:ak}, and the remaining ANO utility after deduction of the subsidy (identical for both ANOs). Total utility increases slightly as capacity grows but is not highly sensitive due to the low cost of storage, especially at the shared CO. Higher capacity leads to smaller shadow prices and a greater share of the gains for the ANOs. The shares do not reflect the 0.5 ratio $r_{ak}$ when capacities are low since shadow prices then contribute significantly to subsidies but not to overall utility. CP $2$ has higher demand than CP $1$ and therefore yields greater cost savings. This is why its subsidy is higher. 

\section{Related Work}
The timely survey  by Paschos \etal on the role of caching in networks usefully highlights the significant infrastructure economies realizable by trading off cache memory for access network bandwidth \cite{paschos2018role}. Cited works by Borst \etal \cite{borst2010distributed} and Poularakis \etal \cite{poularakis2016complexity} on optimal content placement in cache hierarchies are particularly relevant. However, this prior work does not explain how the optimal tradeoff can be realized when CPs, who have exclusive knowledge of demand, are hardly motivated to use this to reduce the costs incurred by the network operator.

The business relation between ISPs and CPs is well-known to be problematic and has given rise to much research and discussion, notably on the issue of network neutrality. The 2009 analysis of the content delivery two-sided market by Musacchio \etal \cite{musacchio2009} is still relevant today while forthcoming network ``cloudification'' does not appear to bring obvious simplifications (\eg, Tang and Ma \cite{tang2019regulating} and Hu \etal \cite{hu2019media}). In our work we suppose content placement takes place under the presently dominant, one-sided pricing model where end-user charges pay for access network costs. CPs therefore have little natural incentive to cooperate in minimizing these costs.

There is relatively little work that seeks to estimate the quantitative value of the memory bandwidth tradeoff. Kelly and Reeves \cite{kelly2001optimal} and Cidon \etal \cite{cidon2002optimal} made early contributions brought up to date by Erman \etal \cite{erman2011cache} and Roberts and Sbihi \cite{roberts2013exploring}. The present work extends the analysis of Elayoubi and Roberts \cite{elayoubi2015performance} for an isolated cache by considering a hierarchy of cache locations. More significantly, we additionally design pricing and value sharing mechanisms that incite network operators and CPs  to cooperatively realize the optimal tradeoff.

Game theory has been widely used to analyse the economic relation between ISPs and CPs. Much of this literature makes the assumption that CPs have an incentive to cache contents close to end-users in the form of improved QoE, either through reduced latency or enhanced throughput. This is the case of recent work by Gourdin \etal \cite{gourdin2017economics}, Mitra \etal \cite{mitra2017emerging}, Mitra and  Sridhar \cite{mitra2019case} and Douros \etal  \cite{douros2017caching}, for example. We disagree that QoE is a sufficiently discriminating criterion since latency in the access network is hardly impacted by negligible differences in propagation time while throughput and storage access times can and should be controlled by adequate provisioning. 
%The network operator is indeed assumed to offer its paying users satisfactory quality as part of its service level agreement, however the cost of adequate provisioning is shared in the two-sided market. 
We deduce the need to introduce subsidies, in the form of side-payments from network operators to CPs, in order to realize the significant savings brought by optimizing the tradeoff. Game theory is used to determine how the savings should be shared between CPs and network operators.

When storage capacity or network bandwidth is limited, it is necessary to partition the resource between multiple CPs. Optimal partitioning  of a limited capacity cache was considered by Dehghan \etal \cite{dehghan2019sharing} and Araldo \etal \cite{araldo2018caching}. Both papers propose iterative schemes where partitions are adjusted by the network operator based on observed per-CP performance. In our network model, optimal cache partitioning is realized by iteratively adjusting the price of storage relative to the price of bandwidth. The CP subsidy is such that they have an incentive to store the overall most popular contents, up to the capacity limit. Dynamic pricing for bandwidth sharing is the basis of network utility maximization (NUM) as introduced by Kelly \etal \cite{kelly1998rate} and applied, in particular, to ISP--CP interaction  by  Hande \etal \cite{hande2009network}. NUM is realized through a Lagrangian decomposition of the optimization into interacting user and network problems. In our proposal the tradeoff optimization is split between CPs placing content to optimize their subsidy while network operators fix prices to maximize utilization of their limited capacity pre-installed resources.

\section{Conclusion}

We have identified the need for network operators to financially reward CPs for placing content in access network caches thus realizing an advantageous memory for bandwidth tradeoff. The considerable potential savings from an optimal tradeoff are currently not realized since the CPs, who exclusively possess the necessary detailed demand data, have no natural incentive to use this to make the optimal placement.  In our proposal, each ANO would give a subsidy to each CP that is proportional to the realized savings. The actual proportion, between 0 and 100\%, would be determined bilaterally.

We designed a value sharing scheme where CPs maximize their subsidy by optimally placing content items in an access network cache hierarchy rooted at the central office. We determined a value distribution that specifies how the cost of the CO cache and corresponding CP subsidy should be divided among  multiple ANOs. This distribution is independent of the individually negotiated proportions of savings handed over to the CP. 

We proposed a distributed iterative approach to optimize the tradeoff based on Lagrangian decomposition. CPs compute their optimal content placements in the cache hierarchy, given unit bandwidth and storage prices, while ANOs fix shadow prices to maximize utilization of limited capacity resources. Neutrality is assured since prices are applied uniformly and content items are placed optimally independently of the CP to which they belong.

The present cache subsidy proposal is novel and considerable scope to extend this preliminary analysis remains. Our simple network model excludes certain aspects of existing and future access networks that impact the tradeoff. For instance, overlapping base station coverage areas or a topology with cross links would allow cooperative caching \cite{krolikowski2018decomposition}. Infrastructure sharing between ANOs in 5G networks would imply further cost and subsidy partitions \cite{andrews2014}.

Our analysis and algorithm design can certainly be improved in several directions. We have not thoroughly evaluated the convergence speed and optimality gap of the proposed distributed optimisation. The subsidy has been designed so that CPs maximize their gain by  making optimal placements but it remains to fully evaluate scope for gaming the system by either ANOs or CPs.  Fallback actions should be defined when some ANO or CP does not participate or behaves irrationally.

Last but not least, a practical cache subsidy scheme needs to be acceptable to ANOs and CPs. Network operators already complain that CPs do not pay them sufficiently and will not at all like the idea of payments going in the opposite direction. CPs, on the other hand, are quite happy with their lucrative business models and may not see any pressing need to change. However, optimizing the network infrastructure is a worthwhile societal goal leading ultimately to lower charges to end users and greater efficiency for both network operators and CPs. Cache subsidies are a viable means to achieve this goal. 
% or


%\appendices

%\section{}
\ifCLASSOPTIONcompsoc
  \section*{Acknowledgments}
\else
  \section*{Acknowledgment}
\fi
The work presented in this article has  benefited from the support of NewNet@Paris, Cisco's Chair ``{\sc Networks for the Future}'' at Telecom ParisTech (\url{http://newnet.telecom-paristech.fr}). Any opinions, findings or recommendations expressed in this material are those of the authors and do not necessarily reflect the views of partners of the Chair.
\bibliographystyle{IEEEtran}
\bibliography{IEEEabrv,refs}

\end{document}